\let\oldvec\vec
\documentclass{article}
\let\vec\oldvec
\usepackage{amssymb}
\usepackage{amsmath}
\usepackage{amsthm}
\usepackage{graphicx}
\usepackage{nicefrac}
\usepackage{subfigure}
\usepackage{xspace}
\usepackage{authblk}

\newtheorem{theorem}{Theorem}[section]
\newtheorem{lemma}[theorem]{Lemma}
\newtheorem{corollary}[theorem]{Corollary}
\newtheorem{definition}[theorem]{Definition}

\newcommand*{\UB}{$\alpha_\delta^u$\xspace}
\newcommand*{\LB}{$\alpha_\delta^l$\xspace}

\newcommand*{{\BudgetGames}}{\emph{bandwidth allocation games}}
\newcommand*{\BudgetGame}{\emph{bandwidth allocation game}}
\newcommand*{{\BUDGETGAMES}}{\emph{Bandwidth Allocation Games}}

\title{On Existence and Properties of Approximate Pure Nash Equilibria\\in Bandwidth Allocation Games\thanks{This work was partially supported by the German Research Foundation (DFG) within the Collaborative Research Centre ``On-The-Fly Computing'' (SFB 901) and by the EU within FET project MULTIPLEX under contract no.\ 317532.}}
\author{Maximilian Drees}
\author{Matthias Feldotto}
\author{S\"oren Riechers}
\author{Alexander Skopalik}
\affil{Heinz Nixdorf Institute \& Department of Computer Science\authorcr University of Paderborn, Germany}

\begin{document}
\maketitle
\begin{abstract}
 In \emph{bandwidth allocation games} (BAGs), the strategy of a player consists of various demands on different resources.
 The player's utility is at most the sum of these demands, provided they are fully satisfied. 
 Every resource has a limited capacity and if it is exceeded by the total demand, it has to be split between the players.
 Since these games generally do not have pure Nash equilibria, we consider
 approximate pure Nash equilibria, in which no player can improve her utility by more than some fixed factor $\alpha$ through unilateral strategy changes.
 There is a threshold $\alpha_\delta$ (where $\delta$ is a parameter that limits the demand of each player on a specific resource) such that $\alpha$-approximate pure Nash equilibria always exist for $\alpha \geq \alpha_\delta$, but not for $\alpha < \alpha_\delta$.
 We give both upper and lower bounds on this threshold $\alpha_\delta$ and show that the corresponding decision problem is ${\sf NP}$-hard.
 We also show that the $\alpha$-approximate price of anarchy for BAGs is $\alpha+1$. 
 For a restricted version of the game, where demands of players only differ slightly from each other (e.g. symmetric games),
 we show that approximate Nash equilibria can be reached (and thus also be computed) in polynomial time using the best-response dynamic.
 Finally, we show that a broader class of utility-maximization games (which includes BAGs)
 converges quickly towards states whose social welfare is close to the optimum.
\end{abstract}

\section{Introduction}
Nowadays, as cloud computing and other data intensive applications such as video streaming gain more and more importance, the amount of data processed in networks and compute centers is growing.
Moore's law for data traffic~\cite{mooresLawForDataTraffic02} states that the overall data traffic doubles each year.
This yields unique challenges for resource management, particularly bandwidth allocation.
As technology cannot follow up with the data increase, bandwidth constraints are often a bottleneck of current systems.

In our paper, we cope with the problem that service providers often cannot satisfy the needs of all customers. That is, the overall size of connections between the provider and all customers exceeds the amount of data that the provider can process.
By allowing different link sizes in network structures, connections between providers and customers with different capacities can be modeled. 
In case a provider cannot fulfill the requirements of all customers, the available bandwidth needs to be split.
This results in customers not being supplied with their full capacity.
In video streaming, for example, this may lead to a lower quality stream for certain customers.
In our setting, we assume that each customer can choose the service providers she wants to use herself.
While this aspect has recently been studied from the compute center's point of view~\cite{Brinkmann:2014},
our work considers limited resources from the customers' point of view.

We study this scenario in a game theoretic setting called {\BudgetGames}.
Here, we are interested in the effects of rational decision making by individuals.
In our context, the customers act as the players.
In contrast, we view the service providers as resources with a limited capacity.
Each possible distribution of a player among the resources (which we view as network entrance points) is regarded as one of her strategies.
Now, each player strives to maximize the overall amount of bandwidth that is supplied to her.
Our main interest lies in states in which no customer wants to deviate from her current strategy, as this would yield no or only a marginal benefit under the given situation.
These states are called (approximate) pure Nash equilibria.
Instead of a global instance enforcing such stable states, they occur as the result of player-induced dynamics.
At every point in time, exactly one player changes her strategy such that the
amount of received bandwidth is maximized, assuming the strategies of the other players are fixed.
We show that if we allow only changes which increase the received bandwidth by some constant factor,
this indeed leads to stable states.
We further analyze the quality of such states in regard to the total bandwidth received by all players
and compare it to the state which maximizes this global payoff.
\paragraph{Related Work.}
Bandwidth allocation games can be considered to be a generalization of market sharing games~\cite{goemans2006market}, in which players choose a set of market in which they offer a service. Each market has a fixed cost and each player a budget. The set of markets a player can service is thus determined by a knapsack constraint.
The utility of a player is the sum of utilities that she receives from each market that she services. Each market has a fixed total profit or utility that is evenly distributed among the players that service the market.

The utility functions of bandwidth allocation games are more general. In particular the influence of a player on the utility share others players receive is not uniform. Players with high demand have a much stronger influence on the bandwidth other players receive than player with small demands. This feature can also be found in demanded congestion games~\cite{Milchtaich1996111}.
Players in a congestion game choose among subsets of resources while trying to minimize costs. The cost of a player is sum of the costs of the resources.
In the undemanded version which was introduced by Rosenthal~\cite{Rosenthal73} the cost of each resource depends only on the number of players using that resource. In the demanded version each player has a demand and the cost of a resource is a function of the sum of demands of the players using the resource. 
In both model the cost caused by a resource is identical for each player that uses the resource. 
In the variant of player-specific congestion games, each player has her own set of cost functions~\cite{Milchtaich1996111} for each resource that map from the number of players using a resource to the cost incurred to that specific player
Mavronicolas et al. \cite{DBLP:conf/mfcs/MavronicolasMMT07} combined these two variations into
demanded congestion games with player-specific constants,
in which the cost functions are based on abelian group operations.
Harks and Klimm~\cite{DBLP:conf/tark/HarksK11} introduced a model in which each player not only picks a subset of resources, but also her single demand on them.
A higher demand equals a higher utility for each player, but also increases the congestion at the chosen resources. The final payoff results from the difference between utility and congestion.

Both, market sharing games and congestion games always posses pure Nash eqilibria. Moreover they are potential games~\cite{monderer1996potential} which implies that every finite sequence of best response dynamics is guaranteed to converge to a pure Nash equilibrium. 
demanded congestion games are potential games only if the cost function are linear or exponential functions~\cite{harks2010existence}. 
For demanded and player-specific games the existence of pure Nash equilibria
this is guaranteed for the special case in which the strategy spaces of the players for the bases of a matroid~\cite{DBLP:journals/tcs/AckermannRV09}. 

Fabrikant et al. \cite{DBLP:conf/stoc/FabrikantPT04} showed that the problem of computing a pure Nash equilibrium is {\sf PLS}-complete. This result implies that the improvement path could be exponentially long. In the case of demanded~\cite{dunkel2008complexity} or player-specific~\cite{ackermann2008complexity} congestion games it is {\sf NP}-hard to decide if there exists a pure Nash equilibrium. 
These negative computational and existence results lead to the study of $\alpha$-approximate pure Nash equilibria which are states in which no player can increases her utility (or decrease her cost) by a {\em factor} of more than $\alpha$.
Chien and Sinclair~\cite{DBLP:journals/geb/ChienS11} showed that in 
symmetric undemanded congestions games and under a mild assumption on the cost functions every sequence of $(1+\varepsilon)$-improving steps 
convergence to $(1+\varepsilon)$-approximate equilibria in polynomial time in the number of players and $\varepsilon^{-1}$.  This result cannot be generalized to asymmetric games as Skopalik and V\"ocking~\cite{DBLP:conf/stoc/SkopalikV08} showed that the problem is still {\sf PLS}-complete.
However, for the case of linear or polynomial cost function Caragiannis et al. presented~\cite{DBLP:conf/focs/CaragiannisFGS11} an algorithm to compute approximate pure Nash equilibria in polynomial time which was slightly improved in  \cite{Skopalik14}.

For demanded congestion games it was shown that $\alpha$-approximate pure equilibria with small values of $\alpha$ exist~\cite{DBLP:conf/approx/HansknechtKS14} and that they can be computed in polynomial time~\cite{caragiannis2012approximate} albeit only for a larger values of $\alpha$. Chen and Roughgarden \cite{DBLP:journals/mst/ChenR09} proved the existence of approximate equilibria in network design games with demanded players.
The results have been used by Christodoulou et al. \cite{DBLP:journals/algorithmica/ChristodoulouKS11} to give tight bounds on the price of anarchy and price of stability of approximate pure Nash equilibria in undemanded congestion games.

To quantify the inefficiency of equilibrium outcomes the price of anarchy has been thoroughly analyzed for exact equilibria
for undemanded \cite{DBLP:journals/siamcomp/AlandDGMS11,DBLP:conf/stoc/ChristodoulouK05,DBLP:conf/stoc/Roughgarden09} as well as for demanded congestion games \cite{DBLP:journals/siamcomp/AlandDGMS11,DBLP:journals/siamcomp/AwerbuchAE13,DBLP:journals/teco/BhawalkarGR14,DBLP:conf/stoc/ChristodoulouK05}.
Christodoulou et al. \cite{DBLP:journals/algorithmica/ChristodoulouKS11} also investigated the PoA for approximate pure Nash equilibria.

Recent work bounded the convergence time to states with a social welfare close to the optimum rather than equilibria.
The concept of smoothness was first introduced by Roughgarden~\cite{DBLP:conf/stoc/Roughgarden09}.
Several variants such as the concept of semi-smoothness~\cite{DBLP:conf/sigecom/LucierL11} followed.
Awerbuch et al. \cite{DBLP:conf/sigecom/AwerbuchAEMS08} proposed $\beta$-niceness which was reworked in~\cite{DBLP:conf/ijcai/AugustineCEFGS11}.
It is the basis of the concept of nice games introduced in~\cite{Anshelevich13}, which we use in our work.
\paragraph{Our Contribution.}
We introduce the notion of \emph{$\delta$-share bandwidth allocation games} (BAGs).
The demand on a resource may not exceed that resource's capacity by a factor of more than $\delta$.
Building on a result from our previous paper \cite{Drees14}, we show that no matter how small we choose $\delta$,
these games generally do not have pure Nash equilibria.
We then turn to $\alpha$-approximate pure Nash equilibria, in which no player can improve her utility by a
factor of more than $\alpha$ through unilateral strategy changes.
We are interested in the threshold $\alpha_{\delta}$ (based on a given $\delta$),
such that for all $\alpha < \alpha_{\delta}$, there is a $\delta$-share BAG without an $\alpha$-approximate pure Nash equilibrium,
and for all $\alpha \geq \alpha_{\delta}$, every $\delta$-share BAG has an $\alpha$-approximate pure Nash equilibrium.
By using a potential function argument, we give both upper and lower bounds for $\alpha_\delta$.
For a general $\delta$-share BAG $\mathcal{B}$ and $\alpha < \alpha_\delta^l$, it is 
{\sf NP}-complete to decide if $\mathcal{B}$ has an $\alpha$-approximate pure Nash equilibria
and {\sf NP}-hard to compute it, if available.
On the other hand, for $\alpha \geq \alpha_\delta^u$ and if the difference between the most-profitable strategies of the players
can be bounded by some constant $\lambda$, then an $(\alpha+\varepsilon)$-approximate Nash equilibrium can be computed efficiently.
We give an almost tight bound of $\alpha+1$ for the $\alpha$-approximate price of anarchy for BAGs and
finally show that utility-maximization games with certain properties converge
quickly towards states with a social welfare close to the optimum.
We then adapt this general result to $\delta$-share BAGs.
\section{Model and Preliminaries}
\label{sec:model}
A \emph{bandwidth allocation game} (BAG) $\mathcal{B}$ is a tuple
$\left(\mathcal{N}, \mathcal{R}, (b_r)_{r \in \mathcal{R}}, (\mathcal{S}_i)_{i \in \mathcal{N}}\right)$
where the set of players is denoted by $\mathcal{N} = \{1,\ldots,n\}$,
the set of resources by $\mathcal{R} = \{r_1,\ldots,r_m\}$, 
the \emph{capacity} of resource $r$ by $b_r$
and the strategy space of player $i$ by $\mathcal{S}_i$.
Each $s_i \in \mathcal{S}_i$ has the form $(s_i(r_1),\ldots,s_i(r_m)) \in \mathbb{R}_{\geq 0}^m$,
with $s_i(r_j) \in \mathbb{R}_{\geq 0}$ being the \emph{demand} of $s_i$ on the resource $r_j$.
We say that a strategy $s_i$ \emph{uses} a resource $r_j$ if $s_i(r_j) > 0$.
$\mathcal{S} = \mathcal{S}_1 \times \ldots \times \mathcal{S}_n$ is the set of strategy profiles
and $u_i: \mathcal{S} \rightarrow \mathbb{R}_{\geq 0}$ denotes the private utility function player $i$ strives to maximize.
For a strategy profile ${\sf s} = (s_1,\ldots,s_n)$,
let $u_{i,r}({\sf s}) \in \mathbb{R}_{\geq 0}$ denote the utility of player $i$ from resource $r$, which is defined as
$$u_{i,r}({\sf s}) := \min\left(s_i(r), \frac{b_r \cdot s_i(r)}{\sum_{j \in \mathcal{N}} s_j(r)} \right).$$
The total utility of $i$ is then defined as $u_i({\sf s}) := \sum_{r \in \mathcal{R}} u_{i,r}({\sf s})$.

Let $\delta > 0$.
We call a bandwidth allocation game a \emph{$\delta$-share bandwidth allocation game}
if for every strategy $s_i$ and every resource $r$, the restriction $s_i(r) \leq \delta b_r$ holds.

Let ${\sf s}$ be an arbitrary strategy profile and $i \in \mathcal{N}$.
We denote with ${\sf s}_{-i} := (s_1,\ldots,s_{i-1},s_{i+1},\ldots,s_n)$ the strategy vector of all players except $i$.
For any $s_i \in \mathcal{S}_i$, we can extend this to the strategy profile $({\sf s}_{-i},s_i) := (s_1,\ldots,s_{i-1},s_i,s_{i+1},\ldots,s_n)$.
We denote with $s^b_i \in \mathcal{S}_i$ the \emph{best response} of $i$ to ${\sf s}_{-i}$
if $u_i({\sf s}_{-i},s^b_i) \geq u_i({\sf s}_{-i},s_i)$ for all $s_i \in \mathcal{S}_i$.

Let $\alpha \geq 1$ and $s_i$ a strategy of player $i$.
If there is a strategy $s_i' \in \mathcal{S}_i$ with 
$\alpha \cdot u_i({\sf s}_{-i},s_i) < u_i({\sf s}_{-i},s_i')$,
then we call the switch from $s_i$ to $s_i'$ an \emph{$\alpha$-move}.
For $\alpha = 1$, we simply use the term \emph{move}.
A strategy profile ${\sf s}$ is called an $\alpha$-approximate pure Nash equilibrium ($\alpha$-NE)
if $\alpha \cdot u_i({\sf s}) \geq u_i({\sf s}_{-i},s_i')$ for every $i \in \mathcal{N}$ and $s_i' \in \mathcal{S}_i$.
For $\alpha = 1$, ${\sf s}$ is simply called a pure Nash equilibrium (NE).
If a bandwidth allocation game eventually reaches an ($\alpha$-approximate) pure Nash equilibrium after a finite number of ($\alpha$-)moves from any initial strategy
profile ${\sf s}$, we say that the game has the \emph{finite improvement property}.

The \emph{social welfare} of a strategy profile ${\sf s}$ is defined as $u({\sf s}) = \sum_{i \in \mathcal{N}} u_i({\sf s})$.
Let $opt$ be the strategy profile with $u(opt) \geq u({\sf s})$ for all ${\sf s} \in \mathcal{S}$.
If $\mathcal{S}^\alpha \subseteq \mathcal{S}$ is the set of all $\alpha$-approximate pure Nash equilibria in a bandwidth allocation game $\mathcal{B}$,
then $\mathcal{B}$'s \emph{$\alpha$-approximate price of anarchy} ($\alpha$-PoA) is the ratio $\max_{{\sf s} \in \mathcal{S}^\alpha} \frac{u(opt)}{u({\sf s})}$.
Again, we simply use the term \emph{price of anarchy} (PoA) for $\alpha = 1$.

Throughout the paper, we are going to use a potential function $\phi: \mathcal{S} \rightarrow \mathbb{R}$
to analyze the properties of bandwidth allocation games.
Let $T_r({\sf s}) := \sum_{i \in \mathcal{N}} s_i(r)$ be the total demand on resource $r$ under strategy profile ${\sf s}$.
We define $\phi({\sf s}) := \sum_{r \in \mathcal{R}} \phi_r({\sf s})$ with 
$$\phi_r({\sf s}) := \left\{\begin{array}{cl} T_r({\sf s}) & \mbox{ \ \ \ if } T_r({\sf s}) \leq b_r \\ b_r + \int_{b_r}^{T_r({\sf s})} \! \frac{b_r}{x} \, \mathrm{d}x  & \mbox{ \ \ \ else} \end{array}\right.$$
%
%
%
%
\section{Pure Nash Equilibria}
The $\delta$-share BAGs in this paper resemble the standard budget games from our previous work \cite{Drees14} in which $\delta$ was unbounded.
This allowed arbitrarily large demands for the strategies.
In particular, the demand of a strategy on a resource $r$ could exceed the capacity $b_r$. 
In $\delta$-share BAGs, that demand is restricted to the interval $[0,\delta b_r]$ for a fixed $\delta > 0$.
We now show that our previous result concerning
the existence of NE still holds for any restriction on the demands.
\begin{definition}
 \label{def:basicBG}
 Let $\delta > 0$ be arbitrary, but fixed.
 Choose $\gamma,\sigma > 0$ and $n \in \mathbb{N}_0$ s.t. $\gamma < \delta$, $\sigma \leq \delta$ and $n \cdot \sigma + \delta = 1$.
 Let $\mathcal{B}_0$ be a $\delta$-share bandwidth allocation game with
 $|\mathcal{N}_0| = n+2$,
 $\mathcal{R}_0 = \{r_1,r_2,r_3,r_4\}$ resources with capacity $1$ and
 the strategy spaces $\mathcal{S}_1 = \{s_1^1 = (\gamma,\delta,0,0), s_1^2 = (0,0,\delta,\gamma)\}$,
 $\mathcal{S}_2 = \{s_2^1 = (\delta,0,\gamma,0), s_2^2 = (0,\gamma,0,\delta)\}$ and
 $\mathcal{S}_i = \{s_i = (\sigma,\sigma,\sigma,\sigma) \}$ for $i \in \{3,\ldots,n+2\}$.
\end{definition}
The players $i \in \{3,\ldots,n+2\}$ serve as auxiliary players to reduce the available capacity of the resources.
Each can only play strategy $s_i$,
so we focus on the two remaining players 1 and 2,
which we regard as the main players of the game.
In every strategy profile, one of them has a utility of
$u := \frac{\gamma}{\delta + \gamma + n \cdot \sigma} + \delta$
while the other one has a utility of $u' := \frac{\delta}{\delta + \gamma + n \cdot \sigma} + \gamma$.
Assume $\delta \leq 1$.
Since $\delta + \gamma + n \cdot \sigma > 1$, we obtain $\frac{\delta - \gamma}{\delta + \gamma + n \cdot \sigma} < \delta - \gamma$
and therefore $u' < u$.
Since the player with utility $u'$ can always change strategy to swap the two utilities,
$\mathcal{B}_0$ does not have a pure Nash equilibrium.

For $\delta > 1$, we choose $n = 0$ and $\gamma > 1$.
In this case, $u = \frac{\gamma}{\delta + \gamma} + 1$ and $u' = \frac{\delta}{\delta + \gamma} + 1$
with $u < u'$.
Again, the player with the lower utility $u$ can always improve her utility.
The game $\mathcal{B}_0$ for $\delta > 1$ is shown in Figure~\ref{fig:figure}.
We conclude the following result.
\begin{corollary}
 \label{theo:noGeneralNE}
 For every $\delta > 0$, there is a $\delta$-share bandwidth allocation game which does not yield a pure Nash equilibrium.   
\end{corollary}
\begin{figure}[h]
 \centering
 \includegraphics[scale=0.7]{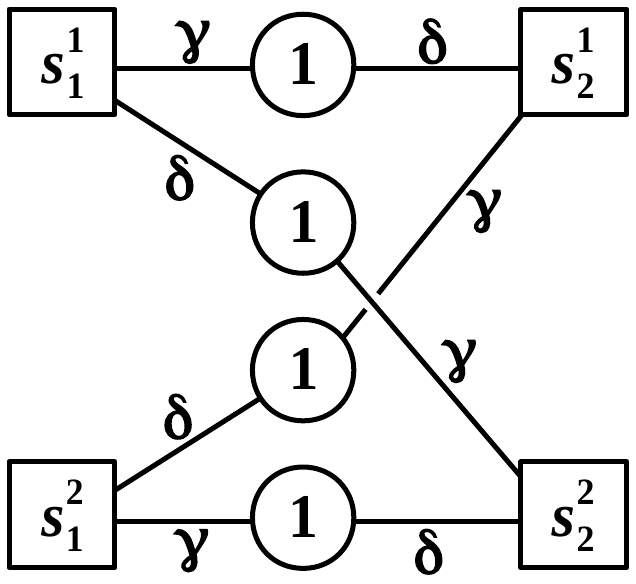}
 \caption{
	   A $\delta$-share bandwidth allocation game without a pure Nash equilibrium for $\delta > 1$.
	   Strategies are depicted as rectangles,
	   resources as circles with their capacity inside. 
	   A link between strategy $s$ and resource $r$ shows that $r$ is used by $s$.
	   The demand of $s$ on $r$ is written next to the link.	   
	 }
 \label{fig:figure}
\end{figure}
%
%
%
%
\section{Approximate Pure Nash Equilibria}
The previous section has shown that we cannot expect any $\delta$-share BAG to have a pure Nash equilibrium.
Therefore, we turn our attention to $\alpha$-approximate pure Nash equilibria.
If $\alpha$ is chosen large enough, any strategy profile becomes an $\alpha$-NE, whereas we know that there may not be an $\alpha$-NE for $\alpha = 1$.
Hence, there has to be a threshold $\alpha_\delta$ for a guaranteed existence of these equilibria in dependency of $\delta$.
In this section, we give both upper and lower bounds on $\alpha_\delta$. 
We start with the upper bound \UB, which we define as follows. 
\begin{definition}
 Let $\delta > 0$. We define the upper bound \UB on $\alpha_\delta$ as 
 $$\text{\UB} := w \cdot \frac{\ln(w)-w + \delta + 1}{\delta} \text{ with  } w = \left(-\frac{1}{2}W_{-1}\left(-2e^{(-\delta) - 2}\right)\right).$$
 \label{def:upperBound}
\end{definition}
Here, $W_{-1}$ is the lower branch of the Lambert W function.
Table~\ref{tab:upperLowerBound} shows a selection of values of \UB.
\begin{theorem}
 \label{theo:lowerBoundExistance}
 Let $\delta > 0$ and $\mathcal{B}$ be a $\delta$-share bandwidth allocation game. 
 For $\alpha \geq$ \UB, $\mathcal{B}$ reaches an $\alpha$-approximate pure Nash equilibrium after a finite number of $\alpha$-moves. 
\end{theorem}
\begin{proof}
 For this proof, we use the potential function $\phi$ introduced in Section~\ref{sec:model}.
 We also need some additional concepts.
 For a resource $r$, let $\phi_r({\sf s}_{-i})$ be the potential of $r$ omitting the demand of player $i$.
 Now, $\phi_{i,r}({\sf s}) := \phi_r({\sf s}) - \phi_r({\sf s}_{-i})$ is the part of $r$'s potential due to strategy $s_i$ if $s_i$ is the last strategy to be considered when evaluating $\phi_r$ (cf. Figure~\ref{fig:approxEqui}).
 Note that we always have $u_{i,r}({\sf s}) \leq \phi_{i,r}({\sf s})$. 
 We are going to show that any strategy change of a player $i$ improving her personal utility by a factor of more than $\alpha$ also results in an increase of $\phi$ if $\alpha$ is chosen accordingly.
 This implies that the game does not possess any cycles and thus always reaches an $\alpha$-NE after finitely many steps (finite improvement property), as the total number of strategy profiles is finite.  
  
 For now, let $\alpha \geq \max_{i,r} \left(\frac{\phi_{i,r}({\sf s})}{u_{i,r}({\sf s})}\right)$ which trivially implies $\phi_{i,r}({\sf s}) \leq \alpha \cdot u_{i,r}({\sf s})\ \forall\ i,r$. 
 Assume that under the strategy profile ${\sf s}$, player $i$ changes her strategy from $s_i$ to $s'_i$, increasing her overall utility by a factor of more than $\alpha$ in the process.
 We denote the resulting strategy profile by ${\sf s}'$. It follows
 \begin{align*}
  \Delta \phi &= \phi({\sf s}') - \phi({\sf s}) = \sum_{r \in \mathcal{R}} \phi_{i,r}({\sf s}') - \sum_{r \in \mathcal{R}} \phi_{i,r}({\sf s}) \geq \sum_{r \in \mathcal{R}} u_{i,r}({\sf s}') - \alpha \cdot \sum_{r \in \mathcal{R}} u_{i,r}({\sf s})   \\
              &= u_i({\sf s}') - \alpha \cdot u_i({\sf s}) > \alpha \cdot u_i({\sf s}) - \alpha \cdot u_i({\sf s}) = 0
 \end{align*}
 Therefore, the potential $\phi$ of $\mathcal{B}$ indeed grows with every $\alpha$-move.
 \begin{figure}[h]
  \centering
  \includegraphics[scale=0.35]{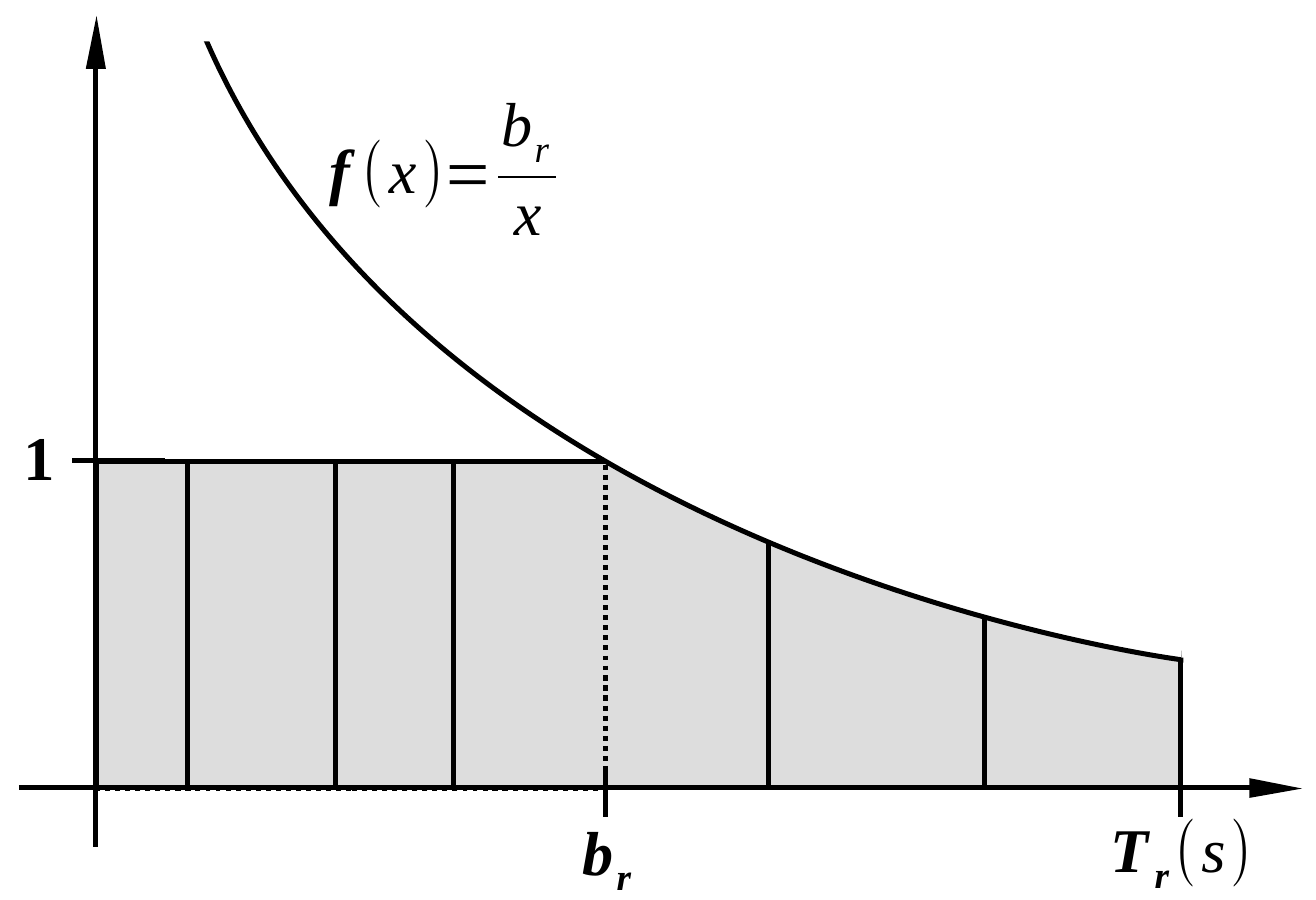}
  \includegraphics[scale=0.35]{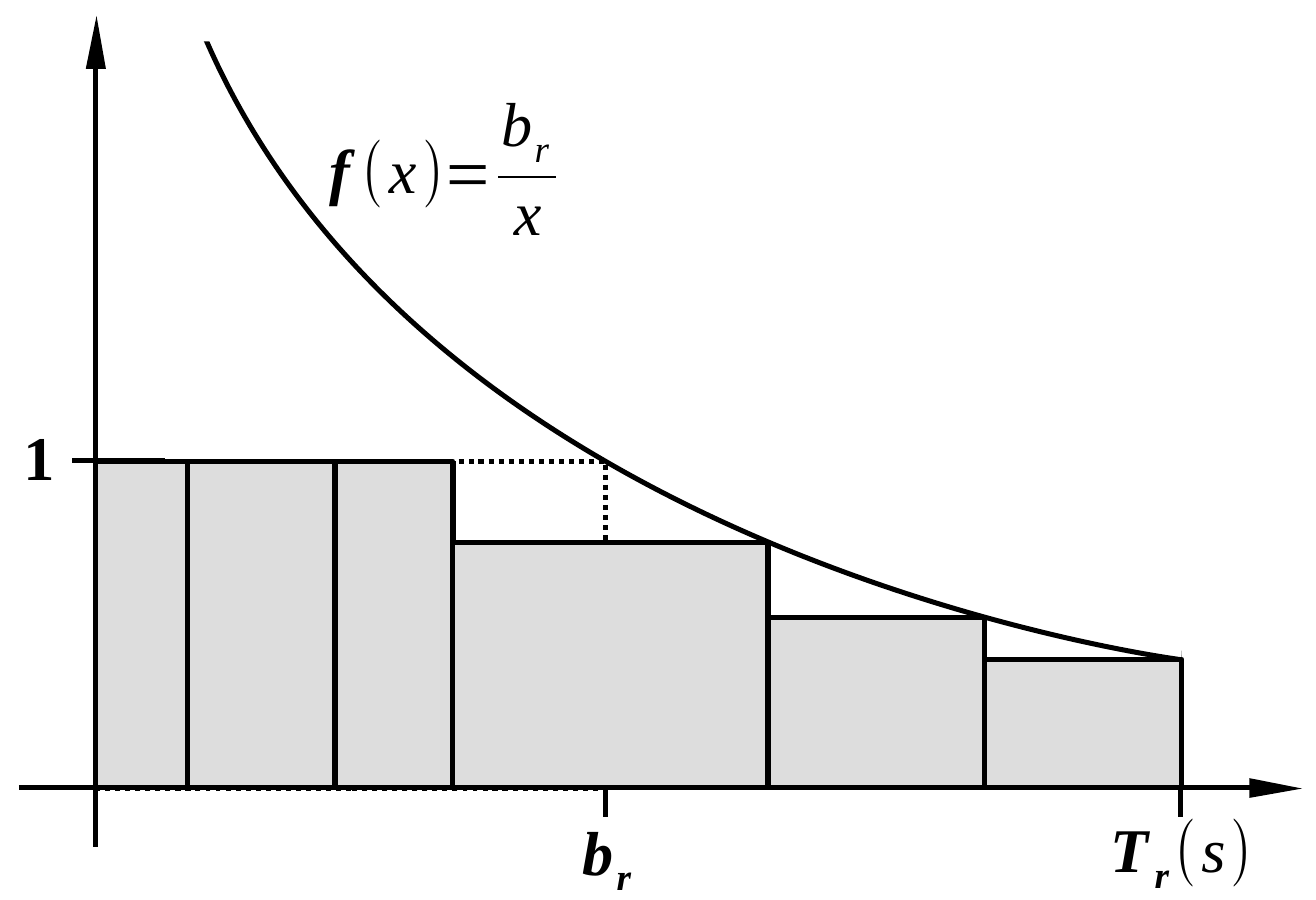}
  \caption{
            The left side shows the potential of resource $r$, divided over the players.
            Each block represents a player currently using $r$.
            The order of the players does not affect the potential as a whole, but the amount caused by each individual player.
            The right side shows how much utility each player initially receives if they arrive at $r$ according to their order.
            Therefore, the utility of the last player in the graph is her actual utility from $r$ under the strategy profile ${\sf s}$.
          }          
  \label{fig:approxEqui}  
 \end{figure} 
 
 
 It remains to be shown that \UB$\geq\max_{i,r} \left(\frac{\phi_{i,r}({\sf s})}{u_{i,r}({\sf s})}\right)$. 
For a resource $r$, define $T_{-i,r}({\sf s}) := T_{r}({\sf s}) - s_i(r)$ as the total demand on $r$ excluding player $i$.
 When the situation is clear from the context, we also write $t_{-i}$ instead of $T_{-i,r}({\sf s})$.
 We make a case distinction based on the size of $t_{-i}$ and look at the two cases $t_{-i} < b_r$ and $t_{-i} \geq b_r$.
 We start with the first one. 
 Note that we can assume $t_{-i} + s_i > b_r$, because otherwise the ratio between potential and utility of $i$ at $r$ would be 1.
 The ratio looks as follows:
 \begin{align*}
  \frac{\phi_{i,r}({\sf s})}{u_{i,r}({\sf s})} &= \frac{b_r-t_{-i} + \int_{b_r}^{t_{-i}+s_i} \frac{b_r}{x} \, \mathrm{d}x}{\frac{b_r \cdot s_i}{t_{-i}+s_i}} \\
                                   &= (t_{-i}+s_i) \cdot \frac{b_r-t_{-i} + b_r \cdot \ln(\frac{t_{-i}+s_i}{b_r})}{b_r \cdot s_i} 
 \end{align*} 
 First we show that this ratio does not decrease as $s_i$ grows larger.
 Its derivative by $s_i$ is
 $$\frac{\partial}{\partial \, s_i} \, \frac{\phi_{i,r}({\sf s})}{u_{i,r}({\sf s})} = \frac{b_r(s_i - t_{-i}) - b_r \cdot t_{-i} \cdot \ln\left(\frac{t_{-i} + s_i}{b_r}\right) + t_{-i}^2}{b \cdot s_i^2}$$  
 The numerator can be bounded below by 
 \begin{align*}
       & \ b_r(s_i - t_{-i}) - b_r \cdot t_{-i} \cdot \ln\left(\frac{t_{-i} + s_i}{b_r}\right) + t_{-i}^2 \\
     = & \ b_r(s_i - t_{-i}) - b_r \cdot t_{-i} \cdot \ln\left(1+ \frac{t_{-i} + s_i - b_r}{b_r}\right) + t_{-i}^2 \\
  \geq & \ b_r(s_i - t_{-i}) - b_r \cdot t_{-i} \cdot \frac{t_{-i} + s_i - b_r}{b_r} + t_{-i}^2 \\
     = & \ b_r \cdot s_i - s_i \cdot t_{-i} = s_i (b_r - t_{-i}) \geq 0
 \end{align*}
 and therefore, the original ratio becomes only worse for bigger values of $s_i$. 
 So from now on, we substitute $s_i$ by its upper bound $\delta b_r$.
 Next we determine the worst-case value for $t_{-i}$.
 The derivative by $t_{-i}$ is
 $$ \frac{\partial}{\partial \, t_{-i}} \, \frac{\phi_{i,r}({\sf s})}{u_{i,r}({\sf s})} = \frac{b_r \cdot \ln\left(\frac{t_{-i} + \delta b_r}{b_r}\right) + 2b_r - 2t_{-i} - \delta b_r}{\delta b_r^2}$$ 
 We are interested in the zero of this function.
 \begin{align*}
                  & b_r \cdot \ln\left(\frac{t_{-i} + \delta b_r}{b_r}\right) + 2b_r - 2t_{-i} - \delta b_r = 0 \\
  \Leftrightarrow & \ln\left(\frac{t_{-i} + \delta b_r}{b_r}\right) = 2\frac{t_{-i}}{b_r} + \delta - 2 \\
  \Leftrightarrow & \ \frac{t_{-i} + \delta b_r}{b_r} = e^{2\frac{t_{-i}}{b_r}} \cdot e^{\delta - 2} \\
  \Leftrightarrow & \ (-2)\frac{t_{-i}}{b_r} - 2\delta = (-2) e^{2\frac{t_{-i}}{b_r}} \cdot e^{\delta - 2} \\
  \Leftrightarrow & \left(-2\frac{t_{-i}}{b_r} - 2\delta \right) \cdot e^{(-2)\frac{t_{-i}}{b_r} - 2\delta} = (-2)e^{(-\delta) - 2} \\
  \Leftrightarrow & \ (-2)\frac{t_{-i}}{b_r} - 2\delta = W_{-1}\left(-2e^{(-\delta) - 2}\right) \\
  \Leftrightarrow & \ t_{-i} = \left(-\frac{1}{2}\right) b_r W_{-1}\left(-2e^{(-\delta) - 2}\right) - \delta b_r \\
  \Leftrightarrow & \ t_{-i} = b_r (w - \delta) \text{ for } w = \left(-\frac{1}{2}W_{-1}\left(-2e^{(-\delta) - 2}\right)\right)
 \end{align*}
 One can show that this function has a zero at $t_{-i} = b_r (w - \delta)$ for $w = \left(-\frac{1}{2}W_{-1}\left(-2e^{(-\delta) - 2}\right)\right)$. 
 $W_{-1}$ denotes the lower branch of the Lambert W function, which is used since $b_r < t_{-i} + s_i = t_{-i} + \delta b_r$ and therefore the value $W = (-2)\frac{t_{-i}}{b_r} - 2\delta < (-2)\frac{b_r (1-\delta)}{b_r} - 2\delta = -2 < -1$.
 Using the obtained values for both $s_i$ and $t_{-i}$, the worst-case ratio between the potential caused at a resource $r$ and the actual utility is
 $$\alpha_\delta^u = w \cdot \frac{\ln(w)-w + \delta + 1}{\delta} \text{ for } w = \left(-\frac{1}{2}W_{-1}\left(-2e^{(-\delta) - 2}\right)\right)$$
 For $t_{-i} > b_r (w - \delta)$, the ratio we seek only becomes smaller as $t_{-i}$ grows.
 This especially holds for $t_{-i} \geq b_r$, when the ratio between potential and utility becomes
 $$\frac{\phi_{i,r}({\sf s})}{u_{i,r}({\sf s})} = \frac{\int_{t_{-i}}^{t_{-i}+s_i} \frac{b_r}{x} \, \mathrm{d}x}{\frac{b_r \cdot s_i}{t_{-i}+s_i}} = (t_{-i} + s_i) \frac{\ln(1 + \frac{s_i}{t_{-i}})}{s_i}$$
 %
%
 %
 By the same methods used above, one can show that this reaches its maximum for $t_{-i} = b_r$ and $s_i = \delta b_r$, 
 which yields 
 $$\frac{\phi_{i,r}({\sf s})}{u_{i,r}({\sf s})} =  (1+\delta) \frac{\ln(1+\delta)}{\delta} \leq \alpha_\delta^u$$
 %
 Therefore, $\alpha_\delta^u$ is indeed the worst-case ratio possible.
\end{proof}
Now that we have an upper bound on $\alpha_\delta$,
we give a lower bound \LB, as well.
\begin{definition}
 Let $\delta > 0$. We define the lower bound \LB on $\alpha_\delta$ as 
 $$\text{\LB} := \frac{2 \sqrt{\delta^2(\delta+2)}+\delta-1}{4\delta-1}$$ 
\end{definition}
Again, we list some values for $\alpha_\delta^l$ in Table~\ref{tab:upperLowerBound}.
\begin{theorem}
 \label{theo:lowerBound}
 Let $\delta > 0$ and $\alpha <$ \LB.
 There is a $\delta$-share bandwidth allocation game without an $\alpha$-approximate pure Nash equilibrium.
\end{theorem}
\begin{proof}
 We refer to the $\delta$-share BAG from Definition \ref{def:basicBG}.
 If we fix $\delta$, the ratio between $u$ and $u'$ becomes a function $f$ in $\gamma$.
 $$f(\gamma) := \frac{\delta + \frac{\gamma}{\delta + \gamma + n \cdot \sigma}}{\gamma + \frac{\delta}{\delta + \gamma + n \cdot \sigma}} = \frac{\gamma + \delta(\delta + \gamma + n \cdot \sigma)}{\delta + \gamma(\delta + \gamma + n \cdot \sigma)} = \frac{\gamma + \delta(\gamma + 1)}{\delta + \gamma(\gamma + 1)}$$
 Deriving $f$ with respect to $\gamma$ yields
 $$f'(\gamma) = \frac{\delta^2 - \delta \gamma (\gamma+2) - \gamma^2}{(\delta + \gamma^2 + \gamma)^2} = 0 \text{ for } \gamma_0 = \frac{\sqrt{\delta^3 + 2 \delta^2}-\delta}{\delta+1}$$
 One can check that this is indeed the only local maximum of $f$ for $\gamma > 0$.
 $$f(\gamma_0) = \frac{2 \sqrt{\delta^2(\delta+2)}+\delta-1}{4\delta-1} = \alpha_\delta^l$$ 
 
\end{proof}
\begin{figure}
 \centering
 \includegraphics[scale=.5]{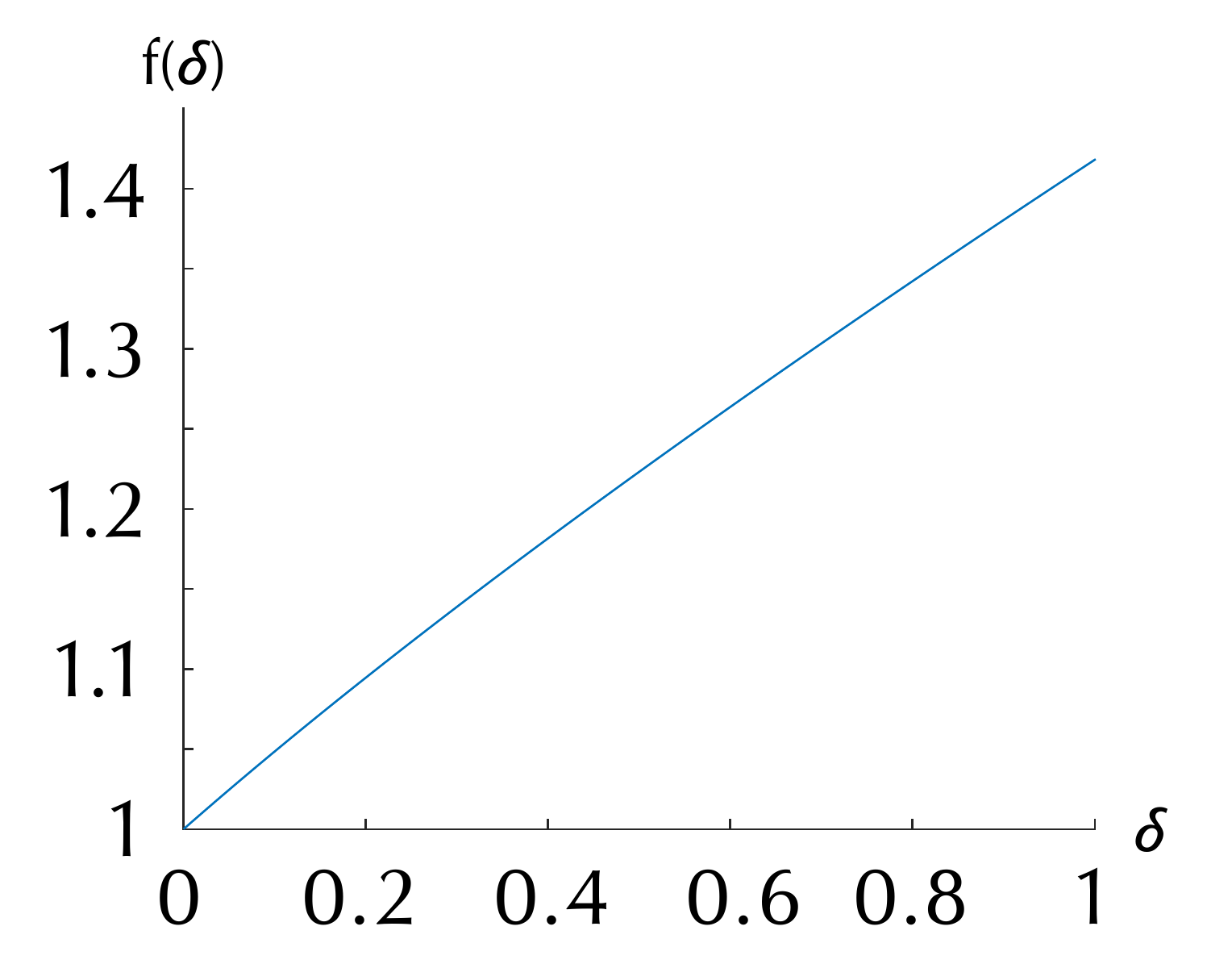}
 \caption{
           Upper bound $f(\delta) := \alpha_\delta^u$ of $\alpha_\delta$ for $\delta \in \ ]0,1]$.           
         }
 \label{fig:plot}
\end{figure}
The smaller $\delta$ is chosen, the better our result, i.e. 
the gap between \UB and \LB becomes smaller and \UB decreases.
While a value of $\delta = 1$ already is a realistic assumption as it states that the demand on a resource may not exceed its capacity,
it also means that one player is able to fully occupy any resource.
However, if we think back to our motivation, it usually takes several thousand clients to exhaust the capacity of a provider.
In this context, \UB -approximate Nash equilibria are close to the definition of (regular) Nash equilibria.
This can also be seen in Figure~\ref{fig:plot}, where $\alpha_\delta^u$ is plotted for $\delta \in \ ]0,1]$.
\begin{table}
 \centering
 \begin{tabular}{|r|r|r|r|r|r|r|r|r|r|r|}
 \hline
 $\delta$ & 0.1 & 0.2 & 0.3 & 0.4 & 0.5 & 0.6 & 0.7 & 0.8 & 0.9 & 1 \\
 \hline
 $\alpha_\delta^u$ & 1.0485 & 1.0946 & 1.1388 & 1.1816 & 1.2232 & 1.2637 & 1.3033 & 1.3422 & 1.3804 & 1.4181 \\
 \hline
 $\alpha_\delta^l$ & 1.0170 & 1.0335 & 1.0497 & 1.0656 & 1.0811 & 1.0964 & 1.1114 & 1.1261 & 1.1405 & 1.1547 \\
 \hline
 \end{tabular}
 \caption{Upper and lower bounds for $\alpha_\delta$ derived from $\delta$.}
 \label{tab:upperLowerBound}
\end{table}

Theorem~\ref{theo:lowerBound} states that for $\alpha$ below \LB,
an $\alpha$-NE cannot be guaranteed in general.
The following result shows that below this lower bound, it is computationally hard to both check for a given $\delta$-share BAG 
whether it has such an equilibrium and to compute it.
\begin{theorem}
 Let $\delta > 0$ and $\alpha <$ \LB.
 Computing an $\alpha$-approximate Nash equilibrium for any $\delta$-share {\BudgetGame} is {\sf NP}-hard.
 \label{theo:NPhard}
\end{theorem}
\begin{proof}
 We reduce from the exact cover by 3-sets problem.
 Given an instance $\mathcal{I} = (\mathcal{U},\mathcal{W})$
 consisting of a set $\mathcal{U}$ with $|\mathcal{U}| = 3m$ and 
 a collection of subsets $\mathcal{W} = \mathcal{W}_1,\ldots,\mathcal{W}_q \subseteq \mathcal{U}$ with $|\mathcal{W}_k| = 3$ for every $k$,
 computing an exact cover of $\mathcal{U}$ in which every element is in exactly one subset is {\sf NP}-hard.
 For $\delta > 0$, we choose an instance $\mathcal{I}$ with $q-m \geq \frac{1}{\delta}$.
 Let $u' := \frac{\delta}{\delta + \gamma + \sigma} + \gamma$.
 
 From $\mathcal{I}$, we create a BAG $\mathcal{B}$ by combining two smaller games $\mathcal{B}_0$ and $\mathcal{B}_{\mathcal{I}}$.
 $\mathcal{B}_0$ is the BAG from Definition \ref{def:basicBG}.
 We label its two main players as player 1 and 2. 
 $\mathcal{B}_{\mathcal{I}}$ is constructed from $\mathcal{I}$ as follows.
 Every subset $\mathcal{W}_i$ is represented by a player $i+2$. 
 Every element $j \in \mathcal{U}$ is represented by a resource $r_j$ with capacity $b_{r_j} = 1$. 
 We assume that $2\delta$ exceeds this capacity, i.e. $2\delta > 1$.
 Otherwise, we refer to Definition \ref{def:basicBG} and add auxiliary players with singular strategy spaces to reduce the available capacity of the resources.
 We also introduce one additional resource $r'$ with $b_{r'} = (q-m) 3 \alpha \delta + \frac{u'}{\alpha}$. 
 The strategy space of each player $i \in \{3,\ldots,q+2\}$ is $\{s_i^1,s_i^2\}$ with
 $$ s_i^1(r_j) := \left\{\begin{array}{cl} \delta & \mbox{ \ \ \ if } j \in \mathcal{W}_{i-2} \\ 0 & \mbox{ \ \ \ else} \end{array}\right. \text{ \ \ \ and \ \ \ } s_i^2(r') = 3 \alpha \delta $$
 All other demands are 0.  
 
 \begin{figure}[h]
 \centering
 \includegraphics[scale=0.6]{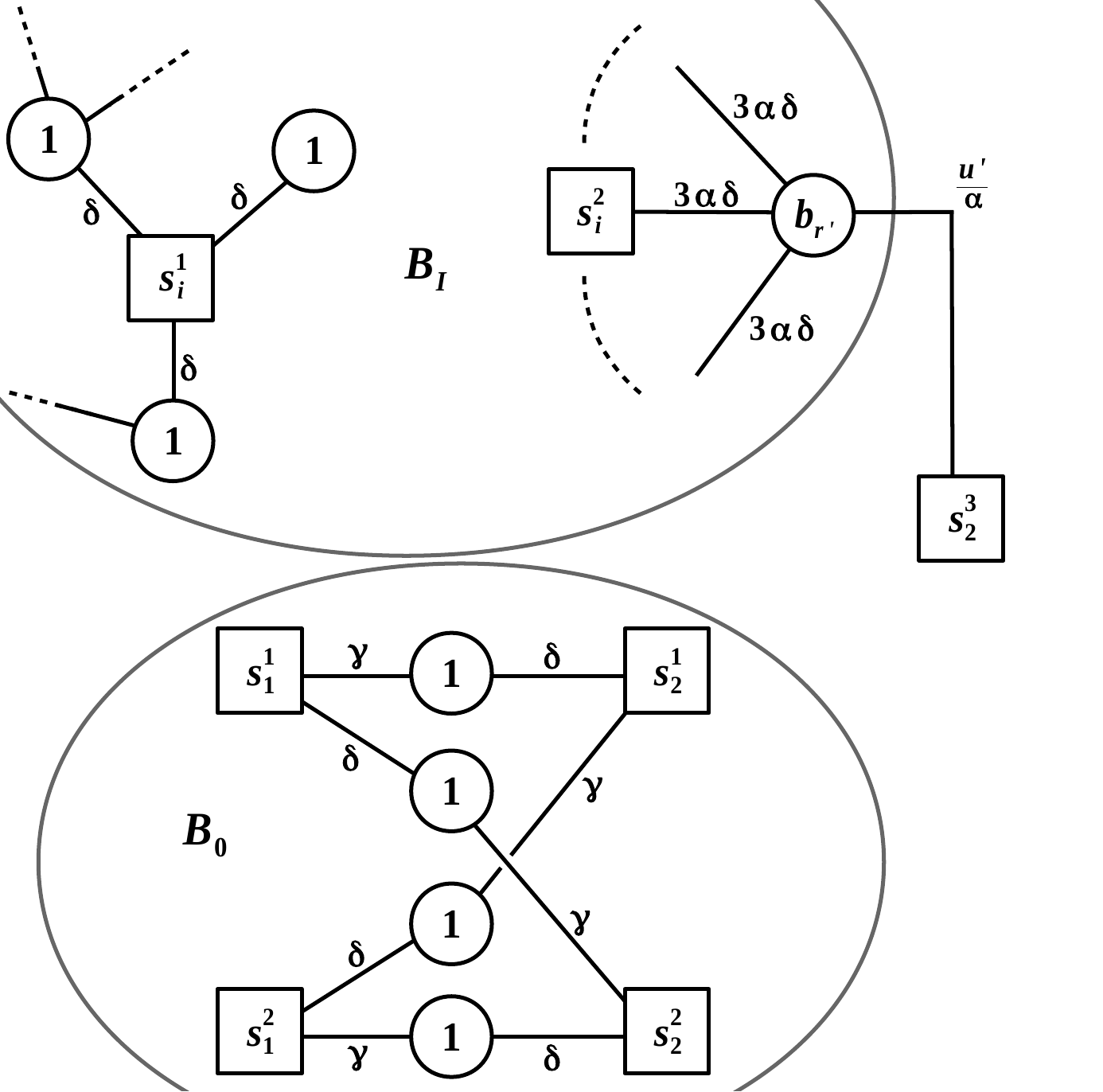}
 \caption{
	   The BAG $\mathcal{B}$, created from $\mathcal{B}_0$ and $\mathcal{B}_\mathcal{I}$.
         }
 \label{fig:NPcomplete}
\end{figure} 
 
 We combine the two unrelated games $\mathcal{B}_{\mathcal{I}}$ and $\mathcal{B}_0$ into one
 by creating the union of the corresponding sets
 and introducing one additional strategy $s_2^3$ for the second player from $\mathcal{B}_0$.
 This strategy uses only the resource $r'$ with $s_2^3(r') = \frac{u'}{\alpha}$.
 For the final result, see Figure \ref{fig:NPcomplete}.
 
 $\mathcal{B}$ is indeed a $\delta$-share BAG.
 For the resource $r'$, note that both $(q-m) 3 \alpha \delta + \frac{u'}{\alpha} \geq (q-m) 3 \alpha \delta \geq \frac{3 \alpha}{\delta} \delta = 3\alpha$
 and $(q-m) 3 \alpha \delta + \frac{u'}{\alpha} \geq 3 \alpha + \frac{u'}{\alpha} \geq \frac{3}{\alpha} + \frac{u'}{\alpha} \geq \frac{2\delta(1-\delta)}{\alpha\delta} + \frac{u'}{\alpha} \geq \frac{u'(1-\delta)}{\alpha\delta} + \frac{u'}{\alpha} = \frac{u'}{\alpha\delta}$,
 so no demand on $r'$ exceeds $\delta b_{r'}$.
 
 We already know that $\mathcal{B}_0$ has no $\alpha$-NE.
 Since the second player now has an additional strategy $s_2^3$,
 this strategy has to be part of any $\alpha$-NE ${\sf s} = (s_1,s_2,\ldots,s_{q+2})$ of $\mathcal{B}$.
 However, if $u_{2,r'}({\sf s}) < s_2(r')$,
 the player will dismiss this strategy and change back to $s_2^1$ or $s_2^2$.
 Therefore, at most $(q-m)$ players $i \in \{3,\ldots,q+2\}$ from $\mathcal{B}_{\mathcal{I}}$ are allowed to play $s_i^2$.
 Every player $i$ with strategy $s_i^1$ and utility below $3 \delta$ will switch to $s_i^2$.
 Therefore, exactly $m$ players pick $s_i^1$ in ${\sf s}$ and they form an exact cover over the resources $r_1,\ldots,r_{3m}$.
\end{proof}

The proof also shows that the decision version of this problem is ${\sf NP}$-complete.
However, for $\alpha \geq$ \UB and if the utilities $u_i^{\text{opt}} := \max_{s_i \in \mathcal{S}_i} \sum_{r \in \mathcal{R}} \min(s_i(r),b_r)$ of the most-profitable strategies of the players do not 
differ too much from each other, approximate Nash equilibria can be computed efficiently.
For example, symmetric games always have this property.

We do not impose any restriction on how much the demands of a single player may deviate from another between her different strategies.
However, we can assume that $u_i^{\text{opt}}$ and the potential utility of any other strategy differ by a factor of at most $n \delta$.
Otherwise, that strategy would never be chosen.
\begin{lemma}
 Let $\mathcal{B}$ be a $\delta$-share BAG. 
 Then $u_i({\sf s}) \geq \frac{u_i^{\text{opt}}}{(n \delta)^2}$ for all players $i \in \mathcal{N}$ and strategy profile ${\sf s}$. 
 \label{lem:restrictedStrategies}
\end{lemma}
\begin{proof}
 Let $s_i^{\text{opt}}$ be the strategy of $i$ associated with $u_i^{\text{opt}}$.
 First, we show that $u_{i,r}({\sf s}_{-i},s_i^{\text{opt}}) \geq \frac{s_i^\text{opt}(r)}{n\delta}$ holds for all ${\sf s}_{-i}$ and $r$.
 If $T_r({\sf s}_{-i},s_i^{\text{opt}}) \leq b_r$, then the claim is true, as $u_{i,r}({\sf s}_{-i},s_i^{\text{opt}}) = s_i^{\text{opt}}(r)$.
 For $T_r({\sf s}_{-i},s_i^{\text{opt}}) > b_r$, $u_{i,r}({\sf s}_{-i},s_i^{\text{opt}}) = \frac{s_i^{\text{opt}}(r) \cdot b_r}{s_i^{\text{opt}} + T_r({\sf s}_{-i})} \geq \frac{s_i^{\text{opt}}(r) \cdot b_r}{n \delta b_r} = \frac{s_i^{\text{opt}}(r)}{n \delta}$.
 By summing up over all resources, we obtain $u_{i}({\sf s}_{-i},s_i^{\text{opt}}) = \sum_{r \in \mathcal{R}} \frac{s_i^{\text{opt}}(r)}{n\delta} \geq \frac{u_i^{\text{opt}}}{n \delta}$.
 So we can assume wlog that for all strategies $s_i \in \mathcal{S}_i$, $\sum_{r \in \mathcal{R}} \min(s_i(r),b_r) \geq \frac{u_i^{\text{opt}}}{n \delta}$.
 Otherwise, the strategy $s_i^{\text{opt}}$ would yield a higher utility in all situations. 
 By the same arguments made above, this implies $u_i({\sf s}_{-i},s_i) \geq \frac{u_i^{\text{opt}}}{(n \delta)^2}$.
 
\end{proof}
We further need an additional lemma to bound the potential of a BAG in respect to its social welfare.
\begin{lemma}
 \label{lemma:socialWelfareLowerBound}
 For any $\delta$-share BAG and any strategy profile ${\sf s}$, $(1+\log(n \delta)) \cdot u({\sf s}) \geq \phi({\sf s})$. 
\end{lemma}
\begin{proof}
 Consider a resource $r$ and let $u_r({\sf s})$ be the total utility obtained from $r$ by all players, i.e. $u_r({\sf s}) := \sum_{i \in \mathcal{N}} u_{i,r}({\sf s})$.
 We show that $(1+\log(n \delta)) \cdot u_r({\sf s}) \geq \phi_r({\sf s})$, which also proves the lemma.
 We assume that $T_r({\sf s}) \geq b_r$, otherwise $\phi_r({\sf s}) = u_r({\sf s})$.
 So $u_r({\sf s}) = b_r$, while
 $\phi_r({\sf s}) = b_r + \int_{b_r}^{T_r({\sf s})} \! \frac{b_r}{x} \, \mathrm{d}x = b_r \left(1 + \int_{b_r}^{T_r({\sf s})} \! \frac{1}{x} \, \mathrm{d}x \right)$.
 Therefore $\frac{\phi_r({\sf s})}{u_r({\sf s})} = 1 + \int_{b_r}^{T_r({\sf s})} \! \frac{1}{x} \, \mathrm{d}x = 1 + \ln(T_r({\sf s})) - \ln(b_r) \leq 1 + \ln(n \delta b_r) - \ln(b_r) = 1 + \ln(n \delta)$ 
\end{proof}
\begin{theorem}
 Let $\mathcal{B}$ be a $\delta$-share BAG for $\delta \leq 1$,
 $\varepsilon > 0$
 and $\lambda \in \ ]0,1]$ such that for all players $i,j$ $u_i^{\text{opt}} \geq \lambda u_j^{\text{opt}}$.
 Then $\mathcal{B}$ reaches an $(\alpha_\delta^u + \varepsilon)$-approximate NE in $\mathcal{O}\left(\log(n) \cdot n^5 \cdot (\varepsilon \lambda)^{-1}\right)$ $(\alpha_\delta^u + \varepsilon)$-moves. 
\end{theorem}
\begin{proof}
 Let $i$ be the player performing an $(\alpha_\delta^u + \varepsilon)$-move under the strategy profile ${\sf s}$, leading to the strategy profile ${\sf s}'$.
 We can bound the increase in the potential:
 \begin{align*}
  \Phi({\sf s}') - \Phi({\sf s}) &\geq \varepsilon \cdot u_i({\sf s}) \overset{(1)}{\geq} \frac{\varepsilon}{(n \delta)^2} \cdot u_i^{\text{opt}} \overset{(2)}{\geq} \frac{\varepsilon \cdot \lambda}{n(n \delta)^2} \cdot u({\sf s}) \overset{(3)}{\geq} \frac{\varepsilon \cdot \lambda}{n (1+\log(n)) (n \delta)^2} \cdot \Phi({\sf s})
 \end{align*} 
 Inequalities $(1)$ and $(3)$ follow by Lemma~\ref{lem:restrictedStrategies} and \ref{lemma:socialWelfareLowerBound} respectively
 while $(2)$ holds due to $u({\sf s}) = \sum_{j \in \mathcal{N}} u_j({\sf s}) \leq \sum_{j \in \mathcal{N}} u_j^{\text{opt}} \leq \frac{n}{\lambda} u_i^{\text{opt}}$.
 For convenience, we define $\beta := \frac{\varepsilon \cdot \lambda}{n (1+\log(n)) (n \delta)^2}$. 
 Assume that we need $t$ steps to increase the potential from $\Phi({\sf s})$ to $2\Phi({\sf s})$.
 Then $\Phi({\sf s}) = 2\Phi({\sf s}) - \Phi({\sf s}) \geq \beta \cdot t \cdot \Phi({\sf s}) \Leftrightarrow t \leq \beta^{-1}$.
 So in order to double the current potential of $\mathcal{B}$, we need at most $\beta^{-1}$ improving moves.
 Therefore, the game has to reach a corresponding equilibrium after at most $\log\left(\frac{\Phi_{\max}}{\Phi_{\min}}\right) \cdot \beta^{-1}$ improving moves,
 with $\Phi_{\max}$ and $\Phi_{\min}$ denoting the maximum and minimum potential of $\mathcal{B}$, respectively.
 Since $\Phi_{\max} \leq \sum_{i \in \mathcal{N}} u_i^{\text{opt}}$ due to $\delta \leq 1$ and $\Phi_{\min} \geq \sum_{i \in \mathcal{N}} \frac{u_i^{\text{opt}}}{(n \delta)^2}$,
 we can bound $\log\left(\frac{\Phi_{\max}}{\Phi_{\min}}\right) \leq (n\delta)^2$.
 
\end{proof}
To conclude this section, we turn towards the quality of $\alpha$-approximate Nash equilibria.
Although no player has an incentive to change her strategy, the social welfare,
which is the total utility of all players combined, may not be optimal.
To express how well Nash equilibria perform in comparison to a globally determined optimal solution,
the price of anarchy has been introduced.
\begin{theorem}
\label{theo:PoA}
 The $\alpha$-approximate price of anarchy of any $\delta$-share bandwidth allocation game is at most $\alpha+1$. 
 For every $\varepsilon > 0$, there is a $\delta$-share bandwidth allocation game with an $\alpha$-approximate Price of Anarchy of $\alpha+1-\varepsilon$.
\end{theorem}
\begin{proof}
 We begin by showing that $\alpha+1$ is an upper bound for the $\alpha$-approximate Price of Anarchy of a {\BudgetGame} $\mathcal{B}$.
 For this, we do not need to consider $\delta$.
 Let $s$ be an $\alpha$-approximate Nash equilibrium of $\mathcal{B}$
 and $opt$ be the strategy profile with the maximum social welfare. 
 We can lower bound the social welfare of $s$ as follows:
 \begin{align}
  \sum_{i \in \mathcal{N}} u_i(s) = & \sum_{r \in \mathcal{R}} \sum_{i \in \mathcal{N}} u_{i,r}(s) \geq \alpha^{-1} \cdot \sum_{r \in \mathcal{R}} \sum_{i \in \mathcal{N}} u_{i,r}(s_{-i},opt_i) \label{eq:Nash} \\
  = & \alpha^{-1} \cdot \sum_{r \in \mathcal{R}} \sum_{i \in \mathcal{N}} \min\left(opt_i(r), \frac{b_r \cdot opt_i(r)}{opt_i(r) + \sum_{i' \ne i} s_{i'}(r)} \right) \label{eq:utilf} \\
  \ge & \alpha^{-1} \cdot \sum_{r \in \mathcal{R}} \sum_{i \in \mathcal{N}} \min\left(opt_i(r), b_r-\sum_{i' \ne i} u_{i',r}(s)\right) \label{eq:newUtilf} \\
  \ge & \alpha^{-1} \cdot \sum_{r \in \mathcal{R}} \sum_{i \in \mathcal{N}} \min\left(u_{i,r}(opt), b_r-\sum_{i' \ne i} u_{i',r}(s)\right) \notag \\
  \ge & \alpha^{-1} \cdot \sum_{r \in \mathcal{R}_1} \sum_{i \in \mathcal{N}} \min\left(u_{i,r}(opt), b_r-\sum_{i' \ne i} u_{i',r}(s)\right)   \notag \\
      & + \alpha^{-1} \cdot \sum_{r \in \mathcal{R}_2} \sum_{i \in \mathcal{N}} u_{i,r}(opt) \label{eq:splitResources} \\                    
  \ge & \alpha^{-1} \cdot \sum_{r \in \mathcal{R}_1} \left(b_r-\sum_{i \in \mathcal{N}} u_{i,r}(s)\right) + \alpha^{-1} \cdot \sum_{r \in \mathcal{R}_2} \sum_{i \in \mathcal{N}} u_{i,r}(opt) \notag \\
  \ge &\alpha^{-1} \cdot \sum_{r \in \mathcal{R}_1} \left( \sum_{i \in N} u_{i,r}(opt) - \sum_{i \in \mathcal{N}} u_{i,r}(s)\right)   \notag \\
      & + \alpha^{-1} \cdot \sum_{r \in \mathcal{R}_2} \sum_{i \in \mathcal{N}} u_{i,r}(opt) \notag \\
  \ge & \alpha^{-1} \cdot \sum_{r \in \mathcal{R}} \sum_{i \in N} u_{i,r}(opt) - \alpha^{-1} \cdot \sum_{r\in\mathcal{R}_1}\sum_{i \in \mathcal{N}} u_{i,r}(s) \notag \\
  \ge & \alpha^{-1} \cdot \sum _{i \in \mathcal{N}} u_i(opt) - \alpha^{-1} \cdot \sum_{r \in \mathcal{R}} \sum_{i \in \mathcal{N}} u_{i,r}(s) \notag \\
  \ge & \alpha^{-1} \sum _{i \in \mathcal{N}} u_i(opt) - \alpha^{-1} \sum_{i \in \mathcal{N}} u_i(s) \label{eq:end}
 \end{align}
 Observe that (\ref{eq:Nash}) follows from the Nash inequality
 and (\ref{eq:utilf}) from the definition of the utility functions.
 In (\ref{eq:newUtilf}), we change how the strategy change from $s_i$ to $opt_i$
 affects the utility of the players.
 We assume that the utilities in $s$ are defined as usual.
 However, a strategy change by player $i$ does not change the utilities of the other players
 (even if they would profit from it).
 In addition, if the remaining capacity of a resource $r$ is less than the demand of $i$ in $opt_i$, 
 player $i$ receives only the remaining capacity
 (that is $b_r-\sum_{i' \ne i} u_{i',r}(s)$).
 Note that within these modified rules, 
 any strategy change by a player $i$ yields at most as much utility as it would in the regular setting.
 In (\ref{eq:splitResources}), we partition $\mathcal{R}$ into $\mathcal{R}_1$ and $\mathcal{R}_2$,
 where $\mathcal{R}_1$ contains all resources for which at least one player evaluates the $\min$ statement to the second expression. 
 Finally (\ref{eq:end}), bringing $- \alpha^{-1} \cdot u(s)$ to the left side of the inequality
 and multiplying both sides with $\alpha$ yields the upper bound of $\alpha + 1$. 
 
 \begin{figure}[ht]
 \centering
 \includegraphics[scale=0.7]{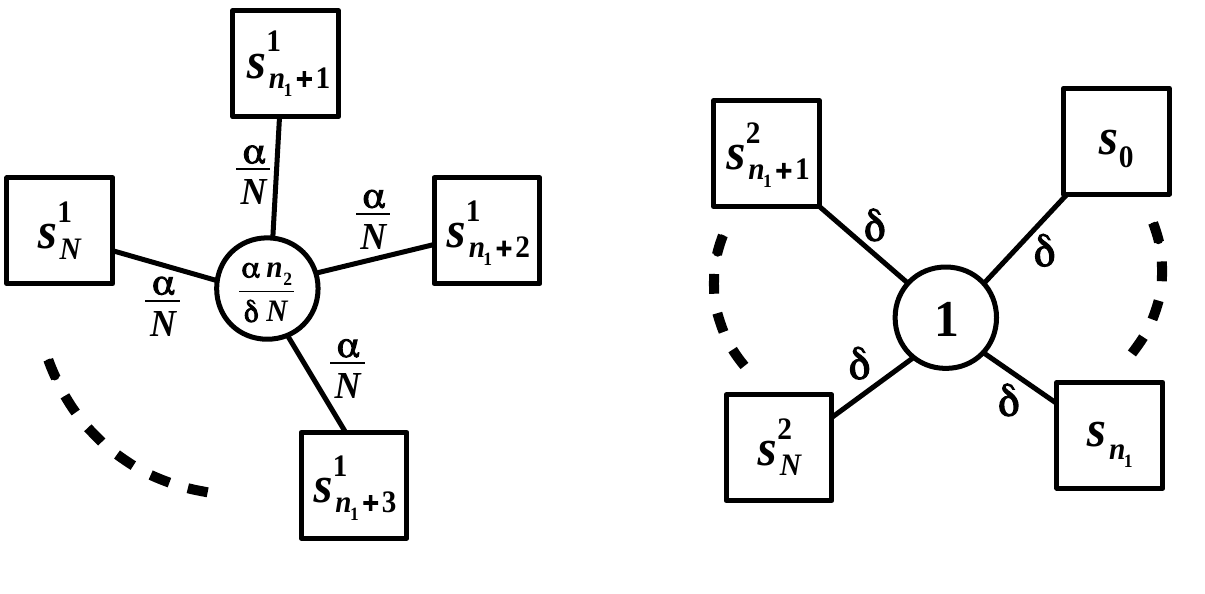}
 \caption{A $\delta$-share BAG with an $\alpha$-approximate price of anarchy close to $\alpha+1$.}
 \label{fig:lowerBoundPoA}
 \end{figure} 
 
 A lower bound of $\alpha$ is trivial for any kind of game and any value of $\alpha$.
 To see that we can come arbitrarily close to $\alpha + 1$, consider the BAG $\mathcal{B}_1$ defined as follows.
 Let both $\delta$ and $\alpha$ be arbitrary, but fixed. 
 Choose $n_1 \in \mathbb{N}$ such that $\delta n_1 \geq 1$ and $n_2 \in \mathbb{N}$.
 $N := n_1+n_2$.
 Define $\mathcal{B}_1$ with 
 $|\mathcal{N}_1| = N$, $\mathcal{R}_1 = \{r_1,r_2\}$ with $b_{r_1} = \frac{\alpha n_2}{\delta N}$ and $b_{r_2} = 1$
 and the strategy spaces $\mathcal{S}_i = \{s_i = (0,\delta)\}$ for $i \in \{1,\ldots,n_1\}$
 and $\mathcal{S}_i = \{s_i^1 = (\frac{\alpha}{N},0), s_i^2 = (0,\delta)\}$ for $i \in \{n_1+1,\ldots,N\}$.
 The resulting game is a $\delta$-share BAG and shown in Figure \ref{fig:lowerBoundPoA}. 
 The players $1,\ldots,n_1$ have only one strategy to choose from.
 Consider the strategy profile ${\sf s} = (s_1,\ldots,s_{n_1},s_{n_1+1}^2,\ldots,s_N^2)$.
 The utility of the players $i = n_1+1,\ldots,N$ is $\frac{1}{N}$ each. 
 Strategy $s_i^1$ yields a fixed utility of $\frac{\alpha}{N}$,
 so while ${\sf s}$ is an $\alpha$-approximate Nash equilibrium with $u({\sf s}) = 1$,
 $opt = (s_1,\ldots,s_{n_1},s_{n_1+1}^1,\ldots,s_N^1)$ has a social welfare of $1 + n_2 \cdot \frac{\alpha}{N}$.
 For $n_2$ large enough, this comes close to $\alpha + 1$.
\end{proof}
Note that Theorem \ref{theo:PoA} even holds for $\alpha <$ \LB, provided the BAG has an $\alpha$-NE.
This result matches our previous work \cite{Drees14}, where we have shown that the price of anarchy for pure Nash equilibria ($\alpha = 1$) is 2. 
\section{Approximating the Optimal Social Welfare}
In this final section, we look at how fast certain utility-maximization games converge towards socially \textit{good} states,
i.e. strategy profiles with a social welfare close to $u(opt)$
if the players keep performing $\alpha$-moves.
We then apply this result to bandwidth allocation games.
For this, we use the concept of nice games introduced in \cite{Anshelevich13}.
A utility-maximization game is \emph{$(\lambda,\mu)$-nice} if for every strategy profile ${\sf s}$,
there is a strategy profile ${\sf s}'$ with
$\sum_{i \in \mathcal{N}} u_i({\sf s}_{-i},s'_i) \geq \lambda \cdot u(opt) - \mu \cdot u({\sf s})$ for constants $\lambda,\mu$.
%
\begin{theorem}
\label{theo:convergence}
 Let $\mathcal{B}$ be a utility-maximization game with a potential function $\phi({\sf s})$
 such that for some $A,B,C \geq 1$, we have that $A \cdot \phi({\sf s}) \geq u({\sf s}) \geq \frac{1}{B} \cdot \phi({\sf s})$,
 $\phi({\sf s}_{-i},s^b_i) - \phi({\sf s}) \geq u_i({\sf s}_{-i},s^b_i) - C \cdot u_i({\sf s})$
 and which is $(\lambda,\mu)$-nice.
 Let $\rho = \frac{\lambda}{C+\mu}$.
 Then, for any $\varepsilon > 0$ and any initial strategy profile ${\sf s}^0$, the best-response dynamic reaches a state ${\sf s}^t$ with 
 $u({\sf s}^t) \geq \frac{\rho(1-\varepsilon)}{AB} u(opt)$
 in at most $\mathcal{O}\left( \frac{n}{A(C+\mu)} \log \frac{1}{\varepsilon} \right)$ steps.
 All future states reached via best-response dynamics will satisfy this approximation factor as well.
\end{theorem}
\begin{proof}
We adapt a modified version of a proof from \cite{Anshelevich13}, in which we do not require an exact potential function.
We assume a specific order in which the players perform their strategy changes.
The next player $i$ is chosen such that she maximizes the term $u_i({\sf s}_{-i},s_i^b) - C \cdot u_i({\sf s})$
under the current strategy profile ${\sf s}$. Then we have
 \begin{align}
  \phi({\sf s}_{-i},s^b_i) - \phi({\sf s}) & \geq u_i({\sf s}_{-i},s^b_i) - C \cdot u_i({\sf s})
			       \geq \frac{1}{n} \left( \sum_{j \in \mathcal{N}} u_j({\sf s}_{-j},s^b_j) - C \cdot u_j({\sf s}) \right) \notag \\
			       & \overset{(1)}{\geq} \frac{1}{n} \left( \lambda \cdot u(opt) - (C + \mu) u({\sf s}) \right) \label{eq:nice} \geq \frac{1}{n} \left( \lambda \cdot u(opt) - A (C + \mu) \phi({\sf s}) \right) =: f({\sf s}) \notag
 \end{align}
 $(1)$ uses the fact that $\mathcal{B}$ is $(\lambda,\mu)$-nice.
 With this definition of $f(s)$, we see that
 \[f({\sf s}) - f({\sf s}_{-i},s^b_i) = \frac{A(C+\mu)}{n} \left( \phi({\sf s}_{-i},s^b_i) - \phi({\sf s}) \right) \geq \frac{A(C+\mu)}{n} f({\sf s})\]
 and therefore $f({\sf s}_{-i},s^b_i) \leq \left( 1 - \frac{A(C+\mu)}{n} \right) f({\sf s})$.
 So if ${\sf s}^0$ is the initial strategy profile, the best response dynamic converges towards a strategy profile ${\sf s}^t$ with
 $$f({\sf s}^t) \leq \left( 1 - \frac{A(C+\mu)}{n} \right)^t f({\sf s}^0)$$
 By setting $t = \left\lceil \frac{n}{A(C+\mu)} \log \frac{1}{\varepsilon} \right\rceil$ and using that
 $\left( 1-\frac{1}{x} \right)^x \leq \frac{1}{e}$, we obtain \\
 $f({\sf s}^t) \leq e^{(\log 1/\varepsilon)^{-1}} \cdot f({\sf s}^0) = \varepsilon \cdot f({\sf s}^0) \leq \varepsilon \cdot \frac{\lambda \cdot u(opt)}{n}$.
 Using these results and the bounds for the potential function, we obtain
 \begin{align}
  u({\sf s}^t) & \geq \frac{1}{B} \phi({\sf s}^t) = \frac{n}{AB(C+\mu)} \cdot \left( \frac{\lambda \cdot u(opt)}{n} - f({\sf s}^t) \right) \notag \\
	 & \geq \frac{n}{AB(C+\mu)} \left( (1-\varepsilon) \frac{\lambda \cdot u(opt)}{n} \right) \geq \frac{\rho(1-\varepsilon)}{AB} u(opt) \notag
 \end{align}
 We also see that $\phi({\sf s}^t) \geq \frac{\rho(1-\varepsilon)}{A} u(opt)$.
 Since $\phi$ grows with every strategy change, 
 this bound also holds for all following strategy profiles.
 
\end{proof}
When adapting this result for bandwidth allocation games, note that the players have to perform $\alpha$-moves when following the best-response dynamic.
Otherwise, we cannot guarantee that $\phi$ is strictly monotone.
\begin{corollary}
 \label{cor:convergenceSpeed}
 Let $\mathcal{B}$ be a $\delta$-share BAG
 and $\alpha \geq$ \UB. 
 For any $\varepsilon > 0$ and any initial strategy profile ${\sf s}^0$, the best-response dynamic using only $\alpha$-moves reaches a state ${\sf s}^t$ with 
 $u({\sf s}^t) \geq \frac{1-\varepsilon}{(\alpha^2+1)(\ln(n\delta)+1)} u(opt)$
 in at most $\mathcal{O}\left( \frac{n}{\alpha+\alpha^{-1}} \log \frac{1}{\varepsilon} \right)$ steps.
 All future states reached via best-response dynamics will satisfy this approximation factor as well.
\end{corollary}
\begin{proof}
 First we show that any $\delta$-share BAG is $(\alpha^{-1},\alpha^{-1})$-nice.
 Let ${\sf s}$ be an arbitrary strategy profile.
 We show that $\sum_{i \in \mathcal{N}} u_i({\sf s}_{-i},s^b_i) \geq \alpha^{-1} \cdot u(opt) - \alpha^{-1} \cdot u({\sf s})$.
 Note that $u_i({\sf s}_{-i},s^b_i) \geq u_i({\sf s}_{-i},opt_i)$ by definition of $s^b_i$.
 This implies $u_i({\sf s}_{-i},s^b_i) \geq \alpha^{-1} \cdot u_i({\sf s}_{-i},opt_i)$ and
 we can therefore copy the proof of Theorem \ref{theo:PoA} to show that $\sum_{i \in \mathcal{N}} u_i({\sf s}_{-i},s^b_i) \geq \alpha^{-1} \cdot u(opt) - \alpha^{-1} \cdot u({\sf s})$.
 
 We now use our potential function $\phi({\sf s})$,
 for which we already know that $\phi({\sf s}) \geq u({\sf s}) \geq \frac{1}{1+\ln(n\delta)} \phi({\sf s})$ (see Lemma \ref{lemma:socialWelfareLowerBound}) and
 $\phi({\sf s}_{-i},s^b_i) - \phi({\sf s}) \geq u_i({\sf s}_{-i},s^b_i) - \alpha \cdot u_i({\sf s})$.
 So we obtain $A = 1$, $B = 1+\ln(n\delta)$ and $C = \alpha$.
 Using these values together with Theorem \ref{theo:convergence} directly leads to our result.
 
\end{proof}
This last result is particular interesting, considering that computing 
the strategy profile $opt$ is {\sf NP}-hard.
Prior to this, an approximation algorithm was only known for games in which the strategy spaces consist of the bases of a matroid over the resources \cite{Drees14}.
Following the best-response dynamic, we can now approximate the optimal solution for arbitrary strategy space structures.
While reaching an actual $\alpha$-approximate NE by this method may take exponentially long, 
we obtain an $\mathcal{O}(\alpha^2 \log(n))$-approximation of the worst-case equilibrium after a linear number of strategy changes.
\bibliographystyle{plain}
\bibliography{references}

\end{document}